\documentclass[conference]{IEEEtran}
\IEEEoverridecommandlockouts

\usepackage{amsmath,amssymb,amsfonts}
\usepackage{algorithm}
\usepackage{algpseudocode}

\usepackage{graphicx}
\usepackage{textcomp}
\usepackage{xcolor}
\def\BibTeX{{\rm B\kern-.05em{\sc i\kern-.025em b}\kern-.08em
    T\kern-.1667em\lower.7ex\hbox{E}\kern-.125emX}}
\usepackage{cite}  

\usepackage{enumerate}
\usepackage{cleveref}
\usepackage{comment}
\usepackage{algorithm}
\usepackage{booktabs}
\usepackage{upgreek}
\usepackage{romannum}

\usepackage{times}
\usepackage{courier}
\usepackage{helvet}
\usepackage[hyphens]{url}
\usepackage{caption}
\usepackage{subcaption}
\usepackage{mathtools}

\usepackage{tikz}
\usetikzlibrary{arrows.meta, positioning, fit, backgrounds, decorations.pathreplacing}

\usepackage{amsthm}
\theoremstyle{plain}
\newtheorem{theorem}{Theorem} 
\newtheorem{lemma}{Lemma}[subsection]

\newtheorem{corollary}{Corollary}

\theoremstyle{definition}
\newtheorem{definition}{Definition}
\newtheorem{example}{Example}
\newtheorem{assumption}{Assumption}[subsection]

\theoremstyle{remark}
\newtheorem{remark}{Remark}[subsection]

\begin{document}
\title{
Bayesian Safety Guarantees for Port-Hamiltonian Systems with Learned Energy Functions
%
}

\author{Chi Ho Leung and Philip E. Par\'{e}*
    \thanks{*Chi Ho Leung and Philip E. Par\'e are with the Elmore Family School of Electrical and Computer Engineering, Purdue University, USA.
    E-mail: leung61@purdue.edu, 
            philpare@purdue.edu. 
    This material is based upon work supported in part by the US National Science Foundation (NSF-ECCS \#2238388).
    }
}

\maketitle

\begin{abstract}
Control barrier functions for port-Hamiltonian systems inherit model uncertainty when the Hamiltonian is learned from data. 
We show how to propagate this uncertainty into a safety filter with independently tunable credibility budgets.
To propagate this uncertainty, we employ a two-stage Bayesian approach.
First, posterior prediction over the Hamiltonian yields credible bands for the energy storage, producing Bayesian barriers whose safe sets are high-probability inner approximations of the true allowable set with credibility $1 - (\eta_{\mathrm{ptB}})$. 
Independently, a drift credible ellipsoid accounts for vector field uncertainty in the CBF inequality with credibility $1 - (\eta_{\rm dr})$. 
Since energy and drift uncertainties enter through disjoint credible sets, the end-to-end safety guarantee is at least $1 - (\eta_{\rm dr} + \eta_{\mathrm{ptB}})$. 
Experiments on a mass-spring oscillator with a GP-learned Hamiltonian show that the proposed filter preserves safety despite limited and noisy observations.
\end{abstract}


\section{Introduction}

Port-Hamiltonian systems (PHS) enjoy a clean separation between what is known from first principles---the interconnection, dissipation, and input structure $(J, R, G)$---and what must be identified from data---the Hamiltonian~$H$ encoding stored energy.
This separation is well exploited in passivity-based control design~\cite{ortega2002interconnection}, where the known structure guides the shaping of a desired energy function, and more recently in structure-preserving learning, where Gaussian process and neural network methods infer~$H$ while respecting passivity and energy balance~\cite{beckers2022gaussian,leung2025learning, zaspel2024data, greydanus2019hamiltonian, duong2024port, desai2021port}.
However, in safety-critical control, the two sides have remained largely disconnected.
Energy-aware barrier constructions assume~$H$ is known and exploit it to restore relative degree one~\cite{singletary2021safety,califano2024effect}, 
while learning-based control barrier function (CBF) methods treat the full drift as an unstructured unknown~\cite{jagtap2020control,long2022safe,dhiman2021control}, spreading uncertainty across the entire vector field even when physical structure could concentrate it in a selected number of scalar energy functions.

The port-transversal barrier framework~\cite{leung2026port} 
takes a first step toward bridging this gap.
It shows that the interconnection topology 
determines which local energy terms must enter a reshaped barrier to restore relative degree one, and under what structural conditions the resulting CBF inequality is feasible assuming bounded inputs.
The classical mechanical energy-aware barrier~\cite{singletary2021safety} emerges as a special case of the port-transversal barrier framework.
However, the barrier cannot be evaluated numerically without the local energy storage functions. 
Since these are scalar functions of the Hamiltonian, they remain unknown until~$H$ is identified.
The central question driving this work is:
\begin{quote}
    \emph{How can we ensure safety when the Hamiltonian of a dynamic system is learned from data?}
\end{quote}
To answer the question, we turn our attention to the structural separation that port-Hamiltonian modeling provides.
The port-transversal barrier determines the barrier's functional form from the known topology, identifying all safety-relevant model dependence to a minimal set of scalar energy functions at input-carrying subsystems.
This identification enables a two-stage Bayesian treatment in which barrier evaluation and drift prediction are addressed through separate credible sets.
First, posterior credible bands on the blanket energy storages produce port-transversal Bayesian (ptB) barriers whose safe sets are high-probability inner approximations of the true allowable set with credibility $1 - \eta_{\mathrm{ptB}}$.
Then, a drift credible ellipsoid separately accounts for vector field uncertainty in the CBF inequality with credibility $1 - \eta_{\mathrm{dr}}$. 
Since the guarantee combines the two events via a union bound, the end-to-end safety probability is at least $1 - (\eta_{\mathrm{dr}} + \eta_{\mathrm{ptB}})$, with each budget tunable at design time.
Finally, the framework is agnostic to the choice of learner. 
It requires only calibrated credible sets on the blanket storages and the drift, which we instantiate here via Gaussian process regression on the Hamiltonian.

\subsection*{Notation}
The real number line is denoted as \(\mathbb{R}\). 
Vectors in \(\mathbb{R}^n\) are column vectors. 
$[A]_{ij}$ denotes the $i,j$ entry of a matrix $A$.
For any truth function $B(x)$, the indicator function $\mathbb I_{B(x)}=1$, if $B(x)$ returns true, and $0$, otherwise.
$C^1(\mathcal{D})$ denotes the class of continuously differentiable functions in the domain $\mathcal{D}\subseteq \mathbb{R}^n$.
The gradient of a scalar function $H$ with respect to $x$ is denoted as $\nabla_x H$.
We write \(\mathbb E[\cdot]\) for expectation.
The normal and Gaussian process distributions are denoted as $\mathcal{N}(\cdot, \cdot)$ and $\mathcal{GP}(\cdot, \cdot)$, respectively.
For a set $\mathcal S$, its interior is denoted by $\operatorname{int}(\mathcal S)$ and its boundary is denoted by $\partial \mathcal S$.

\section{Preliminaries}
\label{sec:preliminaries}

We recall the port-Hamiltonian system model, the structural
assumptions under which the port-transversal barrier framework
operates, and the key results from~\cite{leung2026port} 
that the present work builds upon.

\subsection{Port-Hamiltonian Systems and Structural Assumptions}

Consider $N$ port-Hamiltonian subsystems:
\begin{equation}
    \begin{split}\label{eq:pH-subsystems}
        \dot{x}_i &= \bigl[J_i(x_i) - R_i(x_i)\bigr]\nabla H_i(x_i)
                     + G_i(x_i)\,u_i + w_i,\\
        y_i &= G_i^\top \nabla H_i(x_i),
        \qquad z_i = \nabla H_i(x_i),
    \end{split}
\end{equation}
with state $x_i \in \mathcal{D}_i \subset \mathbb{R}^{n_i}$,
input $u_i \in \mathcal{U}_i \subset \mathbb{R}^{m_i}$, and
output $y_i \in \mathcal{Y}_i \subset \mathbb{R}^{m_i}$,
where the skew-symmetric interconnection matrix
$J_i(x_i) = -J_i(x_i)^\top$, the positive semi-definite dissipation matrix $R_i(x_i) = R_i(x_i)^\top \succeq 0$, and the
input map~$G_i(x_i)$ are known.
By compositionality, port-Hamiltonian subsystems coupled through
a power-preserving 
interconnection again
form a port-Hamiltonian
system~\cite{van2014port}.
The interconnected system from~\eqref{eq:pH-subsystems} is
written as~\cite{neary2023compositional, van2025learning}:
\begin{equation}
    \begin{split}\label{eq:true-dynamics}
        \dot{x} &= \underbrace{\bigl[J(x) - R(x) + Q(x)\bigr]
                   \nabla H(x)}_{f(x)} + G(x)\,u,\\
        y &= G(x)^\top \nabla H(x),
    \end{split}
\end{equation}
where the state, input, and output concatenate those of the subsystems, i.e.,
$x = {\rm vec}(\{x_i\}_{i=1}^N) \in \mathcal{D} \subset \mathbb{R}^{n}$ with $n = n_1 + \dots + n_N$, and similarly for $u \in \mathcal{U}$ and $y \in \mathcal{Y}$.
The matrices $J(x) = {\rm diag}(\{J_i(x_i)\}_{i=1}^N)$,
$R(x)$, and $G(x)$ are block diagonal, with $R$ and $G$ structured analogously.
The Hamiltonian is the total stored energy $H(x) = \sum_{i=1}^{N} H_i(x_i)$.
The matrix $Q(x)$ is skew-symmetric with zero block-diagonal entries and encodes the coupling between subsystems. 


\begin{assumption}[Structural regularity~{\cite[Asm.~III-A.1--2]{leung2026port}}]%
\label{ass:structural}
  The following conditions hold for system~\eqref{eq:true-dynamics}:
  \begin{enumerate}[(a)]

    \item \label{ass:convexity}
      \emph{Convexity.}\;
      Each $H_i : \mathbb{R}^{n_i} \to \mathbb{R}$ is $C^2$
      and strictly convex with a unique minimizer~$x_i^\star$.

    \item \label{ass:local-dep}
      \emph{Local coupling.}\;
      For all subsystem indices $i,j$ and all
      $k \notin \{i,j\}$,
      $\frac{\partial Q_{ij}}{\partial x_k} \equiv 0$.
  \end{enumerate}
\end{assumption}

\subsection{Influence Graph and Node Sets}
As shown in~\cite{leung2026port}, the sparsity patterns of
$J$, $R$, and $G$ can be exploited for safety filter design.
The influence graph organizes these patterns for graph theoretic analysis.
\label{sec:port-transversality}
\begin{definition}[Influence graph~{\cite[Def.~2]{leung2026port}}]%
\label{def:influence-graph}
  The \emph{influence graph}
  $\mathcal{G}(A) = (\mathcal{V}, \mathcal{E}_A)$ has node set
  $\mathcal{V} = \{1,\dots,N\}$ and edge set
  $\mathcal{E}_A = \bigl\{\{i,j\} : A_{ij}(x) \not\equiv 0\bigr\}$.
\end{definition}

Given a smooth safety specification
$\varphi : \mathcal{D} \to \mathbb{R}$ defining the allowable set
$\mathcal{A} = \{x : \varphi(x) \geq 0\}$ with
$\nabla\varphi \neq 0$ on~$\partial\mathcal{A}$, we identify two
distinguished subsets of~$\mathcal{V}$:
\begin{align}
  \text{\emph{Port nodes:}} \quad
  \mathcal{P}
  &= \bigl\{i \in \mathcal{V} : G_i \not\equiv 0 \bigr\},
  \label{eq:port-nodes} \\
  \text{\emph{Barrier nodes:}} \quad
  \mathcal{B}
  &= \bigl\{i \in \mathcal{V} :
     \nabla_{x_i}\varphi \not\equiv 0 \bigr\},
  \label{eq:barrier-nodes}
\end{align}
where $G_i$ denotes the block of rows of~$G$ corresponding to
subsystem~$i$.

\begin{definition}[Barrier-insulating blanket~{\cite[Def.~5]{leung2026port}}]%
\label{def:insulating-blanket}
  The \emph{barrier-insulating blanket} is the vertex boundary
  of~$\mathcal{B}$ in $\mathcal{G}(A)$:
  \begin{equation}\label{eq:blanket}
    \partial\mathcal{B}
    \coloneqq
    \bigl\{
      j \in \mathcal{V}\setminus\mathcal{B}
      : \exists\, i \in \mathcal{B}
        \text{ with } \{i,j\} \in \mathcal{E}_A
    \bigr\}.
  \end{equation}
  The barrier is \emph{port-insulated} if
  $\partial\mathcal{B} \subseteq \mathcal{P}$.
\end{definition}

Port insulation means that every graph neighbor of the
constrained subsystems carries an input port.
As shown in~\cite{leung2026port}, 
this condition is necessary for the
port-transversal barrier to achieve CBF feasibility under
bounded inputs.

\subsection{Port-Transversal Barrier}
With the blanket $\partial\mathcal{B}$ identified and port-insulation verified, a specification $\varphi$ can be reshaped into a barrier of relative degree one by subtracting the shifted local storages of the blanket port nodes.
\begin{definition}[Port-transversal barrier~{\cite[Def.~6]{leung2026port}}]%
\label{def:pt-barrier}
  Let $\varphi$ have barrier nodes $\mathcal{B}$ with
  port-insulated blanket
  $\partial\mathcal{B} \subseteq \mathcal{P}$.
  For weights $\beta_j > 0$ and shaping parameter $\gamma > 0$,
  the \emph{port-transversal barrier} is:
  \begin{equation}\label{eq:pt-barrier-blanket}
    h_\gamma(x)
    \coloneqq
    \varphi(x)
    - \frac{1}{\gamma}
      \sum_{j \in \partial\mathcal{B} \cap \mathcal{P}}
      \beta_j\,\bar{H}_j(x_j),
  \end{equation}
  where
  $\bar{H}_j(x_j) \coloneqq H_j(x_j) - H_j(x_j^\star) \geq 0$
  is the shifted local storage at the port node~$j$.
\end{definition}

The barrier induces the safe set
$\mathcal{S}_\gamma \coloneqq \{x \in \mathcal{D} : h_\gamma(x) \geq 0\}$.
The following lemma collects the key synthesis properties.

\begin{lemma}[Port-transversal barrier
  synthesis~{\cite[Lem.~2]{leung2026port}}]%
\label{lem:pt-synthesis}
  Under Assumption~\ref{ass:structural} with
  $\partial\mathcal{B} \cap \mathcal{P} \neq \emptyset$:
  \begin{enumerate}[\upshape(i)]
    \item \label{item:pt-restored}
      \emph{Port-transversality.}\;
      The input Lie derivative is:
      \begin{equation}\label{eq:Lg-hgamma}
        L_G h_\gamma(x)
        = -\frac{1}{\gamma}
          \sum_{j \in \partial\mathcal{B} \cap \mathcal{P}}
          \beta_j\,\nabla H_j(x_j)^\top G_j(x),
      \end{equation}
      which is nonzero for all
      $x \in \partial\mathcal{S}_\gamma \setminus
      \mathcal{Z}_{\partial\mathcal{B}}$,
      where
      $\mathcal{Z}_{\partial\mathcal{B}}
      \coloneqq \{x : L_G h_\gamma(x) = 0\}$
      is the blanket degeneracy set.
      In particular, $h_\gamma$ has relative degree one on
      $\mathcal{D}\setminus\mathcal{Z}_{\partial\mathcal{B}}$.

    \item \label{item:pt-safe-set}
      \emph{Safe-set inclusion.}\;
      $\mathcal{S}_\gamma \subseteq \mathcal{A}$,
      with
      $\mathcal{S}_{\gamma_1} \subseteq \mathcal{S}_{\gamma_2}$
      for $\gamma_1 \leq \gamma_2$, and
      $\operatorname{int}(\mathcal{A})
      \subseteq \bigcup_{\gamma>0}\mathcal{S}_\gamma
      \subseteq \mathcal{A}$.
  \end{enumerate}
\end{lemma}

Equation~\eqref{eq:Lg-hgamma} reveals that the input sensitivity
of the reshaped barrier depends on the blanket efforts
$\nabla H_j(x_j)$, which vanish only at the energy
equilibria~$x_j^\star$.
The degeneracy set
$\mathcal{Z}_{\partial\mathcal{B}}$ is where the
port-transversal barrier loses first-order input authority.
Managing feasibility near this set is a central concern of~\cite{leung2026port}.

\section{Problem Formulation}
\label{subsec:problem}

We consider the port-Hamiltonian system \eqref{eq:true-dynamics}
where the interconnection and dissipation matrices $J, R$ and the input map~$G$ are known from the physical network topology, while the Hamiltonian~$H^\dagger$, and hence the drift
$f^\dagger(x) \coloneqq [J(x) - R(x)]\nabla H^\dagger(x)$, is unknown.

\vspace{0.5em}
\noindent\textit{Allowable set.}
The user specifies $c \geq 1$ safety specifications
$\{\varphi_s\}_{s=1}^c$, 
defining the allowable set:
\[
    \mathcal{A}
    \coloneqq
    \bigcap_{s=1}^{c}
    \bigl\{x \in \mathcal{D} : \varphi_s(x) \geq 0\bigr\},
\]
encoding safety specifications such as kinematic constraints,
bounds on subsystem's energy, or combinations thereof (see Example~\ref{ex:composite}).
Since specifications $\varphi_s$ may depend on $H$, the allowable set itself is unknown
until $H$ is identified.

\vspace{0.5em}
\noindent\textit{Bayesian model from data.}
We are given dynamic measurements
$\mathfrak{D} = \{(\tilde{x}(t_k), u(t_k))\}_{k=1}^K$,
which can be instantiated as time-stamped
state--derivative--input tuples as
in~\cite{beckers2022gaussian}, or irregularly sampled noisy state
observations as in~\cite{leung2025learning}.
A Bayesian identification procedure, e.g., GP regression with a
port-Hamiltonian kernel, maps $\mathfrak{D}$ to a posterior
$\pi(\cdot \mid \mathfrak{D})$ on a model class $\Theta$,
inducing a family of candidate Hamiltonians
$\{H_\theta\}_{\theta \in \Theta}$, the corresponding drifts
$\{f_\theta\}_{\theta \in \Theta}$, and closed-loop dynamics:
\begin{equation}\label{eq:theta-closed-loop}
    \dot{x} = f_\theta(x) + G(x)\,u(x).
\end{equation}
We interpret the true Hamiltonian $H$ as indexed by an unknown
$\theta^\dagger \in \Theta$ and quantify uncertainty via
the posterior $\pi(\cdot \mid \mathfrak{D})$ and credible sets
$\Theta_\eta \subset \Theta$ as in
Appendix~\ref{sec:bayesian-cbf-theory}.

\vspace{0.5em}
\noindent\textit{ptB-CBF objective.}
Given $\pi(\cdot\! \mid\! \mathfrak{D})$ and credibility levels
$\eta_{\mathrm{dr}}$ and $\eta_{\mathrm{ptB}}$, our objectives are:
\begin{enumerate}[(P1)]

    \item \label{problem:b-barrier}
      \emph{Bayesian barrier design.}\;
      Design port-transversal Bayesian barriers
      $\{h_{\mathrm{ptB},s}\}_{s=1}^c$ whose joint safe set
      $\mathcal{S}_{\mathrm{ptB}}
      = \bigcap_{s=1}^c \{x : h_{\mathrm{ptB},s}(x) \geq 0\}$
      is, with posterior probability at least
      $1 - \eta_{\mathrm{ptB}}$, an inner approximation of the
      true allowable set~$\mathcal{A}$;

    \item \label{problem:vf-guarantee}
      \emph{Safety filter and guarantees.}\;
      Design a safety filter $u^*(x)$ under unknown system drift, such that
      $\mathcal{S}_{\mathrm{ptB}}$ is forward invariant with probability at least $(1 - \eta_{\mathrm{dr}})$,
      yielding an overall safety guarantee of at least
      $1 - (\eta_{\mathrm{dr}} + \eta_{\mathrm{ptB}})$.
\end{enumerate}
\noindent
Problems~(P\ref{problem:b-barrier})
and~(P\ref{problem:vf-guarantee}) are addressed by the Bayesian
constructions of Sections~\ref{sec:posterior-energy-sections}
and~\ref{sec:B-safety-v-ebcbf}, with independent uncertainty
budgets $\eta_{\mathrm{ptB}}$ and $\eta_{\mathrm{dr}}$, respectively.

\section{Port-Transversal Bayesian CBFs} \label{sec:ptB-cbf}
We present the construction of port-transversal Bayesian control barrier functions (ptB-CBFs) here.

\subsection{Port-Transversal Bayesian Barriers}\label{sec:ptB-barriers}

To start, we construct Bayesian counterparts of the port-transversal barrier that account for the unknown Hamiltonian.

\subsubsection{Vector-Valued Barrier Specifications}%
\label{sec:barrier-specs}

We first extend the scalar port-transversal barrier of Definition~\ref{def:pt-barrier} to handle multiple safety specifications simultaneously.

\begin{definition}[Vector-valued port-transversal barrier]%
\label{def:vector-pt-barrier}
Given safety specifications $\varphi_1, \dots, \varphi_c$ on~$\mathcal{D}$, shaping parameters $\gamma_s > 0$, and weights $\beta_{s,j} \geq 0$, $\forall s \leq c$, the \emph{vector-valued port-transversal barrier} $h : \mathcal{D} \to \mathbb{R}^c$ has components:
\begin{equation}\label{eq:vector-pt-barrier}
  h_s(x) \coloneqq \varphi_s(x) - \frac{1}{\gamma_s} \sum_{j \in \partial\mathcal{B} \cap \mathcal{P}} \beta_{s,j}\,\bar{H}_j(x_j), \qquad s = 1,\dots,c.
\end{equation}
The induced safe set is
$\mathcal{S} \coloneqq \bigcap_{s=1}^c \{x \in \mathcal{D} : h_s(x) \geq 0\}.$
\end{definition}

Since $\bar{H}_j \geq 0$ and $\beta_{s,j} \geq 0$, we have $h_s \leq \varphi_s$ pointwise and hence $\mathcal{S} \subseteq \mathcal{A}$. 
Port-transversality of $h_s$ follows from Lemma~\ref{lem:pt-synthesis}\,(\ref{item:pt-restored}) whenever $\beta_{s,j} > 0$ for some $j \in \partial\mathcal{B} \cap \mathcal{P}$. 
For mechanical systems with configuration-dependent inertia, the Hamiltonian is not separable in canonical coordinates. 
Corollary~\ref{cor:mechanical} 
\cite[Cor.~1]{leung2025energy}
restores separability via the screw-momenta lifting, after which Definition~\ref{def:vector-pt-barrier} applies with the kinetic energy $T(q,p)$ as the unique shifted blanket storage.

\begin{example}[Composite energy--kinematic specification]%
\label{ex:composite}
For a fully actuated mechanical system with $\mathcal{B} = \{q\}$ and $\mathcal{P} = \{p\}$, 
Corollary~\ref{cor:mechanical} 
identifies the kinetic energy $T(q,p)$ as the unique shifted blanket storage, even though the Hamiltonian is not separable in canonical coordinates when $M(q)$ depends on configuration. The following four specifications are captured by~\eqref{eq:vector-pt-barrier}:
\begin{equation*}
\begin{aligned}
  h_q &= \bar{h}_q(q) - \tfrac{1}{\gamma}\,T(q,p),
    && \text{with}\ \ \varphi_q = \bar{h}_q(q),\;\; \beta_{q,p} = 1,
  \\
  h_T &= \bar{T} - T(q,p),
    && \text{with}\ \ \varphi_T = \bar{T} - T,\;\; \beta_{T,p} = 0,
  \\
  h_V &= \bar{h}_V(q) - \tfrac{1}{\gamma}\,T(q,p),
    && \text{with}\ \ \varphi_V = \bar{h}_V(q),\;\; \beta_{V,p} = 1,
  \\
  h_H &= \bar{H} - H(q,p),
    && \text{with}\ \ \varphi_H = \bar{H} - H,\;\; \beta_{H,p} = 0.
\end{aligned}
\end{equation*}
The kinematic barriers $h_q$ and $h_V$ require reshaping, i.e.,  $\beta_s > 0$.
In contrast, the energy barriers $h_T$ and $h_H$ are already port-transversal through their dependence on $T$, i.e., $\beta_s = 0$. 
Port transversality of each component follows from 
Corollary~\ref{cor:mechanical} 
and Lemma~\ref{lem:pt-synthesis}\,(\ref{item:pt-restored}). 
For the Bayesian development, all four barriers share the same uncertain quantity---the kinetic energy $T(q,p)$---so a single credible band on $T$ suffices to construct the Bayesian counterpart of the entire vector barrier.
\end{example}

We now turn our attention to constructing posterior predictions for these quantities and use them to produce conservative Bayesian barriers.

\subsubsection{Posterior Energy Storages}
\label{sec:posterior-energy-sections}

Recall that total stored energy $H(x) = \sum_{j=1}^{N} H_j(x_j)$
, that each storage $H_{j}$ is a function of~$x_{j}$ alone.
When the learning algorithm respects this separability---e.g., through an additive kernel
$k(x,x') = \sum_{i} k_i(x_i, x_i')$---each subsystem inherits
its own scalar GP posterior:
  $H_{*,j}(x_{j})
  \sim \mathcal{N}\bigl(
    \mu_{H_{j}}(x_{j}),\;
    \sigma_{H_{j}}^2(x_{j})
  \bigr),$
with $\mu_{H_{j}}, \sigma_{H_{j}}$ obtained from the upstream
learning algorithm, e.g., GP-PHS~\cite{beckers2022gaussian} or
MS-PHS~\cite{leung2025learning}.
Since $x_{j}^\star$ is fixed and known, the shifted storage
$\bar{H}_{j} = H_{j} - H_{j}(x_{j}^\star)$
is itself Gaussian with:
\begin{equation}\label{eq:Ebar-posterior}
  \bar{H}_{*,j}(x_{j})
  \sim \mathcal{N}\bigl(
    \mu_{\bar{H}_{j}}(x_{j}),\;
    \sigma_{\bar{H}_{j}}^2(x_{j})
  \bigr),
\end{equation}
where the mean and variance follow from standard GP conditioning on the subsystem's GP for $H_{j}$.
If we follow the common convention of normalizing the equilibrium energy so that
$H_{j}(x_{j}^\star) = 0$ 
, then
$\mu_{\bar{H}_{j}} = \mu_{H_{j}}$ and
$\sigma_{\bar{H}_{j}} = \sigma_{H_{j}}$
directly.

\begin{remark}[Mechanical systems and posterior kinetic energy]%
  \label{rem:mechanical-posterior}
  For an open-chain rigid-body mechanism, the canonical PHS form
  and its separable Hamiltonian in screw-momenta coordinates are
  established in~\cite[Ch.~2.3, 3.2]{duindam2009modeling}.
  By 
  Corollary~\ref{cor:mechanical},
  the influence graph has
  $\mathcal{B} = \{q\}$, $\mathcal{P} = \{p\}$, $d = 1$, and the
  unique shifted storage is the kinetic energy
  $\bar{H}_p(x) = T(q,p)$.
  In canonical coordinates $(q,p)$, however, $H(q,p)$ is
  \emph{not} separable when the mass matrix $M(q)$ depends on
  configuration $q$.
  The kinetic energy $T(q,p) = H(q,p) - H(q,0)$ must then be
  computed as the difference of two correlated evaluations of a joint GP over the full state, yielding:
  
    \vspace{-3.5ex}
  {\small
  \begin{align}
    \mu_T(q,p) &= \mu_H(q,p) - \mu_H(q,0),
    \label{eq:mu-T}\\[2pt]
    \sigma_T^2(q,p) &= \sigma_H^2(q,p) + \sigma_H^2(q,0)
      - 2\,\mathrm{Cov}\bigl[H_*(q,p),\, H_*(q,0)\bigr].
    \label{eq:sigma-T}
  \end{align}}
  \hspace{-.275em}Equations~\eqref{eq:mu-T}-\eqref{eq:sigma-T} recover the posterior kinetic energy 
  in Section~\ref{sec:app:mech-inst}
  as a special case.
  No additional mechanical modeling assumptions beyond
  Assumption~\ref{ass:structural} are required.
\end{remark}

\subsubsection{Bayesian Barrier Construction}%
\label{sec:bayesian-barrier-construction}
Each scalar component of the vector-valued port-transversal barrier $h_s$ from Definition~\ref{def:vector-pt-barrier} is an instance of the port-transversal form $h_\gamma$ in~\eqref{eq:pt-barrier-blanket}, with a known functional structure but unknown shifted storages $\{\bar{H}_j(x_j)\}_{j \in \partial\mathcal{B} \cap \mathcal{P}}$.
Let $\mu_{\bar{H}_{j}}, \sigma_{\bar{H}_{j}}$ denote the posterior mean and standard deviation of $\bar{H}_{j}$ from \eqref{eq:Ebar-posterior}.

\begin{assumption}[Energy credible band]%
  \label{ass:energy-credible-band}
  For each shifted storage $\bar{H}_{j}$ in the blanket $\partial \mathcal B \cap \mathcal P$,
  there exist $\lambda_{\mathrm{ptB}} > 0$ and
  $\eta_{\mathrm{ptB}} \in (0,1)$ such that:
  {\small
  \begin{equation}\label{eq:energy-band}
  \begin{aligned}
      \pi\Bigl(
      \theta :
      |\bar{H}_{j,\theta}(x_j) - \mu_{\bar{H}_{j}}(x_j)|
      \leq \lambda_{\mathrm{ptB}}\, \sigma_{\bar{H}_{j}}(x_j),
    \Bigm| \mathfrak{D}
    \Bigr)
    \geq 1 - \eta_{\mathrm{ptB}},
  \end{aligned}
  \end{equation}}
  for each $j \in \partial \mathcal B \cap \mathcal P$.
\end{assumption}

While the existence of a suitable $\lambda_{\mathrm{ptB}}$ depends on
the upstream learning algorithm, it is typically straightforward to
verify once the algorithm is fixed.
In the GP-PHS/MS-PHS setting, a pointwise sub-Gaussian
concentration bound yields
$\lambda_{\mathrm{ptB}} = \sqrt{2\ln(1/\eta)}$;
combining this with a finite-cover argument over~$\mathcal{D}$
gives the uniform band in
Assumption~\ref{ass:energy-credible-band}.
For each state $x$, define the credible box over the blanket storages:
\begin{equation}\label{eq:credible-box}
\begin{aligned}
  \mathcal{I}(x)
  \coloneqq
  \prod_{j \in \partial\mathcal{B} \cap \mathcal{P}}
  \bigl[
    \mu_{\bar{H}_{j}}(x_{j})
      - \lambda_{\mathrm{ptB}}\,\sigma_{\bar{H}_{j}}(x_{j}),\;\\
    \mu_{\bar{H}_{j}}(x_{j})
      + \lambda_{\mathrm{ptB}}\,\sigma_{\bar{H}_{j}}(x_{j})
  \bigr].
\end{aligned}
\end{equation}
With the additive storage structure mentioned in
Section~\ref{sec:posterior-energy-sections}, the subsystem posteriors
are independent, so the credible box $\mathcal{I}(x)$ is exactly the joint credible region rather than a conservative outer approximation.
We write $h_s(x_{\mathcal{B}}; E)$ to denote the barrier~\eqref{eq:vector-pt-barrier} with the shifted storages replaced by candidate values $E \in \mathbb{R}^{|\partial\mathcal{B} \cap \mathcal{P}|}$.
Now, we define that port-transversal Bayesian barrier as follows.

\begin{definition}[Port-transversal Bayesian barrier]%
  \label{def:h-eb}
  Let $h$ be a vector-valued port-transversal barrier as in
  Definition~\ref{def:vector-pt-barrier}.
  The \emph{port-transversal Bayesian barrier} is $h_{\mathrm{ptB}}$ with components:
  \begin{equation}\label{eq:heb-def}
    h_{\mathrm{ptB},s}(x)
    \coloneqq
    \inf_{E \in \mathcal{I}(x)}
      h_s(x;\, E),
    \qquad s = 1, \dots, c,
  \end{equation}
  with joint design safe set
  $\mathcal{S}_{\mathrm{ptB}}
  \coloneqq \bigcap_{s=1}^c \{x : h_{\mathrm{ptB},s}(x) \geq 0\}$.
\end{definition}

Let $h_{s,\theta}(x) \coloneqq h_s(x_{\mathcal{B}};\,\{\bar{H}_j(x_j)\}_{j \in \partial\mathcal{B}
    \cap \mathcal{P}})$
denote the $s$-th port-transversal
barrier evaluated at the unknown shifted storages.
Therefore, $E$ in \eqref{eq:heb-def} denotes a candidate of the unknown shifted storages $\{\bar{H}_j(x_j)\}_{j \in \partial\mathcal{B} \cap \mathcal{P}}$ drawn from $\mathcal I(x)$.

\begin{lemma}[ptB barrier credible dominance]%
  \label{lem:pb-cred-dominance}
  Suppose Assumption~\ref{ass:energy-credible-band} holds.
  Then, with posterior probability at least
  $1 - \eta_{\mathrm{ptB}}$,
  \begin{equation}\label{eq:heb-dominance}
    h_{s,\theta}(x) \geq h_{\mathrm{ptB},s}(x)
    \qquad \forall x \in \mathcal{D},
    \quad \forall\, s = 1, \dots, c.
  \end{equation}
  In particular,
  $\mathcal{S}_{\mathrm{ptB}} \subseteq
  \mathcal{S}_\theta \coloneqq
  \bigcap_{s=1}^c \{x : h_{s,\theta}(x) \geq 0\}$.
\end{lemma}

\begin{proof}
  By Assumption~\ref{ass:energy-credible-band}, with probability
  at least $1 - \eta_{\mathrm{ptB}}$, the unknown shifted storages
  $\{\bar{H}_j(x_j)\}_{j \in \partial\mathcal{B}
    \cap \mathcal{P}}$
  belong to
  $\mathcal{I}(x)$ for all $x \in \mathcal{D}$.
  On this event, for each $s = 1, \dots, c$:
  \[
    h_{s,\theta}(x)
    = h_s(x;\, \bar{H}_{s,\theta})
    \geq
    \inf_{E \in \mathcal{I}(x)}
      h_s(x;\, E)
    = h_{\mathrm{ptB},s}(x),
  \]
  where the inequality holds because the unknown energy lies in the feasible set of the infimum.
  The set inclusion follows: $h_{\mathrm{ptB},s}(x) \geq 0$ for all~$s$ implies $h_{s,\theta}(x) \geq 0$ for all~$s \leq c$.
\end{proof}


Since each barrier component \(h_s(x;E)\) in \eqref{eq:vector-pt-barrier} is affine in the shifted storages \(E\), the infimum in \eqref{eq:heb-def} over the credible box \(\mathcal{I}(x)\) is attained at a vertex, and the computational cost per component is a single function evaluation.
In general, when the native Hamiltonian is not separable, and a lifting is used to restore separability, the pulled-back storages may become posterior-correlated, in which case the credible box should be replaced with an ellipsoidal credible region and the vertex evaluation with a trust-region subproblem analogous to \eqref{eq:Phi-lower-qp}.
For mechanical systems, this issue does not arise: the blanket path length is \(d=1\), so there is a single shifted storage \(T(q,p)\) and no inter-subsystem coupling.

\subsection{Safety Guarantees via ptB-CBF}
With the port-transversal Bayesian barrier $h_{\mathrm{ptB}}$ and its credible dominance in hand, it remains to account for uncertainty in the drift and assemble a safety filter with end-to-end probabilistic guarantees.

\subsubsection{Vector Field Posterior and Gradient Credible Ellipsoid} \label{sec:vf-gp-n-grad-cred-ellip}

From the learning stage implied in Section~\ref{subsec:problem}, we model the unknown drift $f^\dagger \eqqcolon f_\theta$ in~\eqref{eq:true-dynamics} with a posterior distribution
whose mean $\mu_f(x)\in\mathbb{R}^n$ and covariance
$\Sigma_f(x)\in\mathbb{R}^{n\times n}$ at each $x$.

\begin{assumption}[Vector field credible ellipsoid]
    \label{ass:drift-ellipsoid}
    There exist $\lambda_f>0$ and $\eta_{\rm dr}\in(0,1)$ such
    that the state-dependent ellipsoid:
    \begin{equation}
        \mathcal F_{\eta}(x)
        \coloneqq
        \Bigl\{
            v\in\mathbb{R}^n :
            (v-\mu_f(x))^\top
            \Sigma_f(x)^{-1} (v-\mu_f(x))
            \le \lambda_f^2
        \Bigr\}
        \label{eq:drift-ellipsoid}
    \end{equation}
    satisfies:
    \begin{equation}
        \pi\Bigl(
            \theta :
            f_\theta(x)\in\mathcal F_\eta(x)\
            \forall x\in\mathcal D
            \Bigm|
            \mathfrak D
        \Bigr)
        \ge 1-\eta_{\rm dr}.
        \label{eq:drift-ellipsoid-prob}
    \end{equation}
\end{assumption}
While the choice of a $\lambda_f$ satisfying Assumption~\ref{ass:drift-ellipsoid} depends on the upstream learning algorithm, it is typically straightforward to verify once the algorithm is fixed.
\begin{remark}
    In the GP-PHS/MS-PHS GP setting, the port-Hamiltonian drift $f_\theta(x)$
    is modeled by a vector-valued GP with posterior mean $\mu_f(x)$ and
    covariance $\Sigma_f(x)$, induced by the PHS kernel
    $k_{\mathrm{phs}}$. 
    For each fixed $x$, the weighted 2-norm
    $(f_\theta(x)-\mu_f(x))^\top\Sigma_f(x)^{-1}(f_\theta(x)-\mu_f(x))$
    is $\chi^2$-sub-Gaussian, so choosing
    $\lambda_f^2 = 2\ln(1/\eta_{\rm dr})$ gives
    $\pi(f_\theta(x)\in\mathcal F_\eta(x)\mid\mathfrak D)\ge 1-\eta_{\rm dr}$.
    Using GP sample-path continuity or a union bound over a finite cover of
    $\mathcal D$, we obtain a uniform-in-$x$ ellipsoidal credible set
    satisfying Assumption~\ref{ass:drift-ellipsoid} for a desired
    credibility level $1-\eta_{\rm dr}$.
\end{remark}

Before diving into the derivation of a drift-side lower bound, we introduce the notion of drift credible model set that will set the stage for discussion in Section~\ref{sec:B-safety-v-ebcbf}.
\begin{definition}[Drift credible model set induced by an ellipsoidal band]
\label{def:Theta-dyn}
Given the state-dependent ellipsoids $\{\mathcal F_\eta(x)\}_{x\in\mathcal D}$ in \eqref{eq:drift-ellipsoid}, define the induced credible model set:
\begin{equation}
    \Theta_{\rm dr}
    \coloneqq
    \Bigl\{\theta\in\Theta:\ f_\theta(x)\in\mathcal F_\eta(x)\ \ \forall x\in\mathcal D \Bigr\}.
    \label{eq:Theta-dyn}
\end{equation}
\end{definition}
Therefore, Assumption~\ref{ass:drift-ellipsoid} implies a drift credible model set $\Theta_{\rm dr}$ such that:
\begin{equation}\label{eq:drift-induced-cred-model-set}
    \pi(\Theta_{\mathrm{dr}}\mid \mathfrak{D}) \geq 1 - \eta_{\mathrm{dr}}.
\end{equation}
Next, we are interested in deriving a high-probability lower bound on the drift-side barrier term.
For a fixed barrier $h_{\mathrm{ptB}}: \mathcal{D} \to \mathbb{R}$, i.e., $c=1$, define the drift-side
CBF term:
\begin{equation}
\begin{aligned}
    \Phi_\theta(x)
    \coloneqq\ &
    L_{f_\theta} h_{\mathrm{ptB}}(x)
    + \alpha\bigl(h_{\mathrm{ptB}}(x)\bigr)\\
    =\ &
    \nabla h_{\mathrm{ptB}}(x)^\top f_\theta(x)
    + \alpha\bigl(h_{\mathrm{ptB}}(x)\bigr),
\end{aligned}
    \label{eq:Phi-theta-def}
\end{equation}
for some extended class-$\mathcal{K}_\infty$ function $\alpha$.
For later use, we isolate a purely state-dependent lower bound.

\begin{definition}[Drift-side ptB-CBF lower bound]
    \label{def:Phi-lower}
    For each $x\in\mathcal D$, define:
    \begin{equation}
        \underline\Phi(x)
        \coloneqq
        \inf_{v\in\mathcal F_\eta(x)}
        \Bigl[
            \nabla h_{\mathrm{ptB}}(x)^\top v
            + \alpha\bigl(h_{\mathrm{ptB}}(x)\bigr)
        \Bigr].
        \label{eq:Phi-lower-qp}
    \end{equation}
\end{definition}

The optimization problem~\eqref{eq:Phi-lower-qp} is a convex
quadratic program (QP) in $v$ with an ellipsoidal constraint.
Writing $v = \mu_f(x) + \delta v$ and using the change of
variables $w \coloneqq \Sigma_f(x)^{-1/2}\delta v$, the ellipsoid constraint in \eqref{eq:drift-ellipsoid}
becomes $\|w\|_2^2\le\lambda_f^2$ and the objective in \eqref{eq:Phi-lower-qp} is linear in
$w$:
\begin{align*}
    \underline\Phi(x) &=
        \inf_{\|w\|_2 \leq \lambda_f}
        \big[
            \nabla h_{\mathrm{ptB}}(x)^\top (\mu_f + \Sigma_f^{1/2}(x)w)
            + \alpha\bigl(h_{\mathrm{ptB}}(x))
        \big]\\
        &= \nabla h_{\mathrm{ptB}}(x)^\top \mu_f(x)+ \alpha\bigl(h_{\mathrm{ptB}}(x)\bigr) + z^*,
\end{align*}
where $z^* \coloneqq \inf_{\|w\|_2 \leq \lambda_f} (\Sigma^{1/2}(x)\nabla h_{\mathrm{ptB}}(x))^\top w$.
Since $\inf_{\|w\|_2 \leq \lambda_f} a^\top w$ has a closed-form solution $-\lambda_f\|a\|_2$ for any fixed $a$, the QP~\eqref{eq:Phi-lower-qp} also has the closed-form solution:
\begin{equation}
\begin{aligned}
    \underline\Phi(x)
    =\ &
    \nabla h_{\mathrm{ptB}}(x)^\top \mu_f(x)
    + \alpha\bigl(h_{\mathrm{ptB}}(x)\bigr)\\
    &- \lambda_f
      \bigl\|\Sigma_f(x)^{1/2}\nabla h_{\mathrm{ptB}}(x)\bigr\|_2.
\end{aligned}
    \label{eq:Phi-lower-closed-form}
\end{equation}
Notice that computing \(\underline{\Phi}(x)\) reduces to a trust-region subproblem over \(\mathcal F_\eta(x)\).
Finally, for $c > 1$, we define $\underline\Phi_s(x)$ by replacing $h_{\mathrm{ptB}}$ with $h_{\mathrm{ptB}, s}$, for $s = 1, \dots, c$, in~\eqref{eq:Phi-lower-closed-form}.

\begin{lemma}[High-probability lower bound on drift-side barrier term]
    \label{lem:Phi-lower-hp}
    Under Assumption~\ref{ass:drift-ellipsoid}, with posterior
    probability at least $1-\eta_{\rm dr}$,
    \begin{equation}
        \underline\Phi(x)
         \le 
        \Phi_\theta(x)
        \qquad
        \forall x\in\mathcal D,
        \ \forall \theta\in\Theta.
        \label{eq:Phi-lower-hp}
    \end{equation}
\end{lemma}

\begin{proof}
    On the event in~\eqref{eq:drift-ellipsoid-prob},
    $f_\theta(x)\in\mathcal F_\eta(x)$ for all $x$. By the
    definition of $\underline\Phi(x)$ as an infimum over
    $\mathcal F_\eta(x)$ we have, for each fixed $x$ and $\theta$,
    \begin{equation*}
        \begin{aligned}
            \underline\Phi(x)
            &=
            \inf_{v\in\mathcal F_\eta(x)}
            \bigl[
                \nabla h_{\mathrm{ptB}}(x)^\top v
                + \alpha(h_{\mathrm{ptB}}(x))
            \bigr]\\
            &\le 
            \nabla h_{\mathrm{ptB}}(x)^\top f_\theta(x)
            + \alpha(h_{\mathrm{ptB}}(x)) = \Phi_\theta(x),
        \end{aligned}
    \end{equation*}
    which yields~\eqref{eq:Phi-lower-hp}.
\end{proof}

\subsubsection{Bayesian Safety via ptB-CBFs} \label{sec:B-safety-v-ebcbf}

Combining~\eqref{eq:Phi-theta-def}
and~\eqref{eq:Phi-lower-hp}, a sufficient condition for the B-CBF inequality:
\begin{equation*}
        \underbrace{\nabla h_{\mathrm{ptB}, s}(x)^\top f_\theta(x)}_{L_{f_\theta} h_{\mathrm{ptB}, s}(x)}
          + \underbrace{\nabla h_{\mathrm{ptB}, s}(x)^\top G(x)}_{L_{G} h_{\mathrm{ptB}, s}(x)}u
        \ge  -\alpha(h_{\mathrm{ptB}, s}(x))
\end{equation*}
to hold for all $\theta$ in the credible set is the ptB-CBF constraint:
\begin{equation}
    \underline\Phi_s(x) + \nabla h_{\mathrm{ptB}, s}(x)^\top G(x)u  \ge  0.
    \label{eq:pb-cbf-constraint}
\end{equation}

Given a nominal controller $u_{\rm nom}:\mathcal D\to\mathcal U$, we define the ptB-CBF safety filter as the solution to the QP:
\begin{equation}
    \begin{aligned}
        u^*(x)
        ={}& \arg\min_{u\in\mathcal U}
            \|u - u_{\rm nom}(x)\|^2 \\
        &\; \text{s.t.}\;\;
            \underline\Phi_s(x)
            + \nabla h_{\mathrm{ptB}, s}(x)^\top G(x)u
             \ge  0,\; s = 1, \dots, c.
    \end{aligned}
    \label{eq:pb-cbf-qp}
\end{equation}
When $c=1$, the QP admits a closed-form solution $u^*$ obtained by
letting $[h_{\mathrm{ptB}},h_{\mathrm{ptB}}]_{GG^\top}(x)
\coloneqq \nabla h_{\mathrm{ptB}}(x)^\top G(x)G(x)^\top\nabla h_{\mathrm{ptB}}(x)\in \mathbb{R}^{c\times c}$
and assuming $[h_{\mathrm{ptB}},h_{\mathrm{ptB}}]_{GG^\top}(x)\neq 0$ whenever the constraint is active. 
Then, the solution to
\eqref{eq:pb-cbf-qp} admits the usual closed form~\cite[Thm.~2]{califano2024effect}:
\begin{equation}
    u^*(x)
    = u_{\rm nom}(x)
       - \mathbb I_{\{\Psi_{\mathrm{ptB}}(x)<0\}} 
         \frac{G(x)^\top\nabla h_{\mathrm{ptB}}(x)}
              {[h_{\mathrm{ptB}},h_{\mathrm{ptB}}]_{GG^\top}(x)} 
         \Psi_{\mathrm{ptB}}(x),
       \label{eq:u-safe-pb}
\end{equation}
where $\Psi_{\mathrm{ptB}}(x)
    \coloneqq \underline\Phi(x)
       + \nabla h_{\mathrm{ptB}}(x)^\top G(x)u_{\rm nom}(x)$,
and $\mathbb I_{\{\Psi_{\mathrm{ptB}}(x)<0\}}$ 
encodes complementary slackness of a single inequality constraint.
The regularity conditions for an $h_{\mathrm{ptB}}$ to yield safety are collected into the following assumption:
\begin{assumption}(ptB-CBF regularity)\label{as:ptB-CBF-reg}
    For each barrier specification $s = 1, \dots, c$:
    \begin{enumerate}[i)]
    \item \textit{Barrier regularity:} $h_{\mathrm{ptB}, s}\in C^1(\mathcal D)$ with $\nabla h_{\mathrm{ptB}, s}(x)\neq 0\ \forall x\in\partial\mathcal S_{\mathrm{ptB}, s}$, where $\mathcal S_{\mathrm{ptB}, s}\coloneqq\{x:h_{\mathrm{ptB}, s}(x)\ge 0\}$.
    \item \textit{Feasibility of the QP solution map $u^*$:} The feedback $u^*(x)$ is locally Lipschitz and satisfies the ptB-CBF constraint \eqref{eq:pb-cbf-constraint} for all $x\in\mathcal D$.
    \end{enumerate}
\end{assumption}

We combine the ptB barrier dominance in
Lemma~\ref{lem:pb-cred-dominance}, the drift-side bound in Lemma~\ref{lem:Phi-lower-hp}, and the generic B-CBF in 
Theorem~\ref{th:B-forward-invar-via-B-cbf} 
to obtain a safety guarantee for the true port-transversal safe sets.
\begin{theorem}[Bayesian safety via ptB-CBFs]
\label{th:bayesian-safety-pb-cbf-clean}
Let each $h_{\mathrm{ptB}, s}$ be a barrier function that meets the ptB-CBF regularity conditions in Assumption~\ref{as:ptB-CBF-reg} with the uncertainty constraints:
\begin{enumerate}[i)]
    \item \textit{Energy uncertainty:} 
    Assumption~\ref{ass:energy-credible-band} holds with $\eta_{\mathrm{ptB}}$;
    \item \textit{Drift uncertainty:} Assumption~\ref{ass:drift-ellipsoid} holds with $\eta_{\rm dr}$.
\end{enumerate}
Then, for any deterministic $x_0\in\mathcal S_{\mathrm{ptB}} = \bigcap_{s=1}^c\{h_{\mathrm{ptB}, s}\geq 0\}$,
\begin{equation}
    \pi\Bigl(
        \theta :
        x^\theta(t;x_0)\in\mathcal S_\theta
        \ \forall t\ge0
        \Bigm| \mathfrak D
    \Bigr)
    \ge  1 - (\eta_{\rm dr} + \eta_{\mathrm{ptB}}),
    \label{eq:pb-cbf-safety-prob}
\end{equation}
where $\mathcal S_\theta\coloneqq\bigcap_{s=1}^c\{x:h_{s,\theta}(x)\ge 0\}$ denotes the port-transversal safe set in Definition~\ref{def:vector-pt-barrier} with the true $\theta$.
\end{theorem}

\begin{proof}[Proof]
Consider any design safe set $\mathcal{S}_{\mathrm{ptB}}$ defined in Definition~\ref{def:h-eb};
we first show that $\mathcal S_{\mathrm{ptB}}$ is forward invariant for every $\theta \in \Theta_{\rm dr}$, with $\pi(\Theta_{\rm dr}\mid \mathfrak{D})\geq 1-\eta_{\rm dr}$.
Recall the drift credible set $\Theta_{\mathrm{dr}}$ defined in \eqref{eq:Theta-dyn},
with $\mathcal F_\eta(x)$ given in~\eqref{eq:drift-ellipsoid}.  
By Assumption~\ref{ass:drift-ellipsoid}, 
$$\pi(\Theta_{\rm dr}\mid\mathfrak D)\ge 1-\eta_{\rm dr}.$$
Fix any $\theta\in\Theta_{\rm dr}$. 
Then, $f_\theta(x)\in\mathcal F_\eta(x)$ for all $x\in\mathcal D$.
Further, under Assumption~\ref{ass:drift-ellipsoid}, Lemma~\ref{lem:Phi-lower-hp} yields $\underline\Phi(x)\le \Phi_\theta(x)$ for all $x\in\mathcal D$ and for all $\theta \in \Theta_{\rm dr}$ with $\pi(\Theta_{\rm dr}\mid \mathfrak{D})\geq 1-\eta_{\rm dr}$.
By substituting \eqref{eq:Phi-lower-closed-form} into \eqref{eq:pb-cbf-qp}, notice that for all $\theta \in \Theta_{\rm dr}$, the ptB-CBF constraint~\eqref{eq:pb-cbf-constraint} enforced by $u^*$ in \eqref{eq:pb-cbf-qp} implies the standard CBF inequality for each $s$:
\begin{equation*}
    L_{f_\theta}h_{\mathrm{ptB},s}(x) + L_Gh_{\mathrm{ptB},s}(x)u^*(x)\ge -\alpha(h_{\mathrm{ptB},s}(x)),
\end{equation*}
for all $x\in\mathcal D$.
Therefore, the ptB-CBF constraint enforces the deterministic CBF inequality uniformly over $\theta \in \Theta_{\rm dr}$.
Theorem~\ref{th:B-forward-invar-via-B-cbf},
$\mathcal S_{\mathrm{ptB}}$ is $(1-\eta_{\rm dr})$-Bayesian forward invariant, 
i.e.,
\[
\pi\bigl(\theta:\ x^\theta(t;x_0)\in\mathcal S_{\mathrm{ptB}}\ \forall t\ge0 \mid \mathfrak D\bigr)
\ge 1-\eta_{\rm dr}.
\]

To lift invariance from $\mathcal S_{\mathrm{ptB}}$ to $\mathcal S_\theta$, we define the energy-dominance credible set as:
\[
\Theta_{\mathrm{ptB}}
\coloneqq
\{\theta\in\Theta\!:\!h_{\theta,s}(x)\ge h_{\mathrm{ptB},s}(x)\ \forall x\in\mathcal D,\forall s\!=\!1,\dots,c\}.
\]
Then, by Lemma~\ref{lem:pb-cred-dominance}, 
we have
$\pi(\Theta_{\mathrm{ptB}}\mid\mathfrak D)\ge 1-\eta_{\mathrm{ptB}}$, and for every $\theta\in\Theta_{\mathrm{ptB}}$,
$\mathcal S_{\mathrm{ptB}}\subseteq \mathcal S_\theta$.
Therefore, for any $\theta\in\Theta_{\rm dr}\cap\Theta_{\mathrm{ptB}}$:
$$x^\theta(t;x_0)\in\mathcal S_{\mathrm{ptB}}\ \ \forall t\geq 0,$$
and hence $x^\theta(t;x_0)\in\mathcal S_\theta\ \forall t\geq 0$.
Moreover,
\begin{equation}\label{eq:set-inclusion}
    \Theta_{\rm dr}\cap\Theta_{\mathrm{ptB}}
    \subseteq
    \{\theta:\ x^\theta(t;x_0)\in\mathcal S_\theta\ \forall t\ge 0\}.
\end{equation}
By the union bound and De Morgan's laws,
\begin{equation}\label{eq:union_bound}
    \pi(\Theta_{\rm dr}\cap\Theta_{\mathrm{ptB}}\mid\mathfrak D) \ge 1-(\eta_{\rm dr}+\eta_{\mathrm{ptB}}).
\end{equation}
Combining \eqref{eq:union_bound} with the set inclusion in \eqref{eq:set-inclusion} yields~\eqref{eq:pb-cbf-safety-prob}.
\end{proof}

\section{Numerical Simulations}\label{sc:experiment}

We examine the ptB-CBF framework on two systems of increasing complexity: a one-degree-of-freedom mass-spring oscillator serving as a linear baseline, and a fully actuated planar two-link manipulator with configuration-dependent inertia.
In both cases, the Hamiltonian is learned from noisy, irregularly sampled state observations via a multi-output GP with a port-Hamiltonian kernel (MS-PHS)~\cite{leung2025learning}.
Kernel hyperparameters and the observation-noise variance are estimated by minimizing the negative log marginal likelihood with automatic relevance determination (ARD) length-scales~\cite{rasmussen2003gaussian, neal2012bayesian}.
All closed-loop trajectories are integrated with a forward Euler scheme.

\subsection{Mass-Spring Oscillator}\label{sec:exp-mass-spring}

We consider an undamped mass-spring system
$\ddot{q} = -(k/m)\,q + (b/m)\,u$
with mass $m = 2.0$, stiffness $k = 0.5$, input gain $b = 1.0$, and damping $d = 0.0$.
The true Hamiltonian is $H(q,p) = p^2/(2m) + k\,q^2/2$.
Training data is generated by integrating the system over $t \in [0, 20]$ from initial condition $x_0 = (1.0,\, 0.0)^\top$ under a sinusoidal forcing $u(t) = \cos(1.7\,t)$, using a fourth-order Runge-Kutta solver with step size $\Delta t = 4 \times 10^{-3}$.
From the first half of the resulting trajectory, $K = 25$ state observations are uniformly sampled and corrupted by i.i.d.\ Gaussian noise with standard deviation $\sigma_x = 0.05$.
A third-order Adams-Bashforth GP integrator is then trained for $200$ iterations.

Two barrier specifications are imposed simultaneously:
a kinematic constraint $\varphi_q(x) = q + 1 \geq 0$ (i.e., $q \geq -1$), and an energy upper-bound constraint $\varphi_H(x) = \bar{H} - H(q,p) \geq 0$ with $\bar{H} = 0.05$.
Following Example~\ref{ex:composite}, the kinematic barrier requires reshaping by the kinetic energy ($\beta_{q,p} = 1$), while the energy barrier is already port-transversal ($\beta_{H,p} = 0$).
Both ptB barriers share the same scalar uncertain quantity---the posterior Hamiltonian---so a single credible band suffices.
The nominal controller is $u_{\mathrm{nom}}(t) = 3.0\,\cos(0.85\,t)$, chosen to violate both constraints in the absence of filtering.
The closed-loop filtered trajectory is generated for $T = 20\,\mathrm{s}$ with $N = 500$ Euler steps.

\vspace{.5em}
\noindent\textit{Effect of Credibility Budgets}:
\begin{figure}[t]
    \centering
    \includegraphics[width=\linewidth]{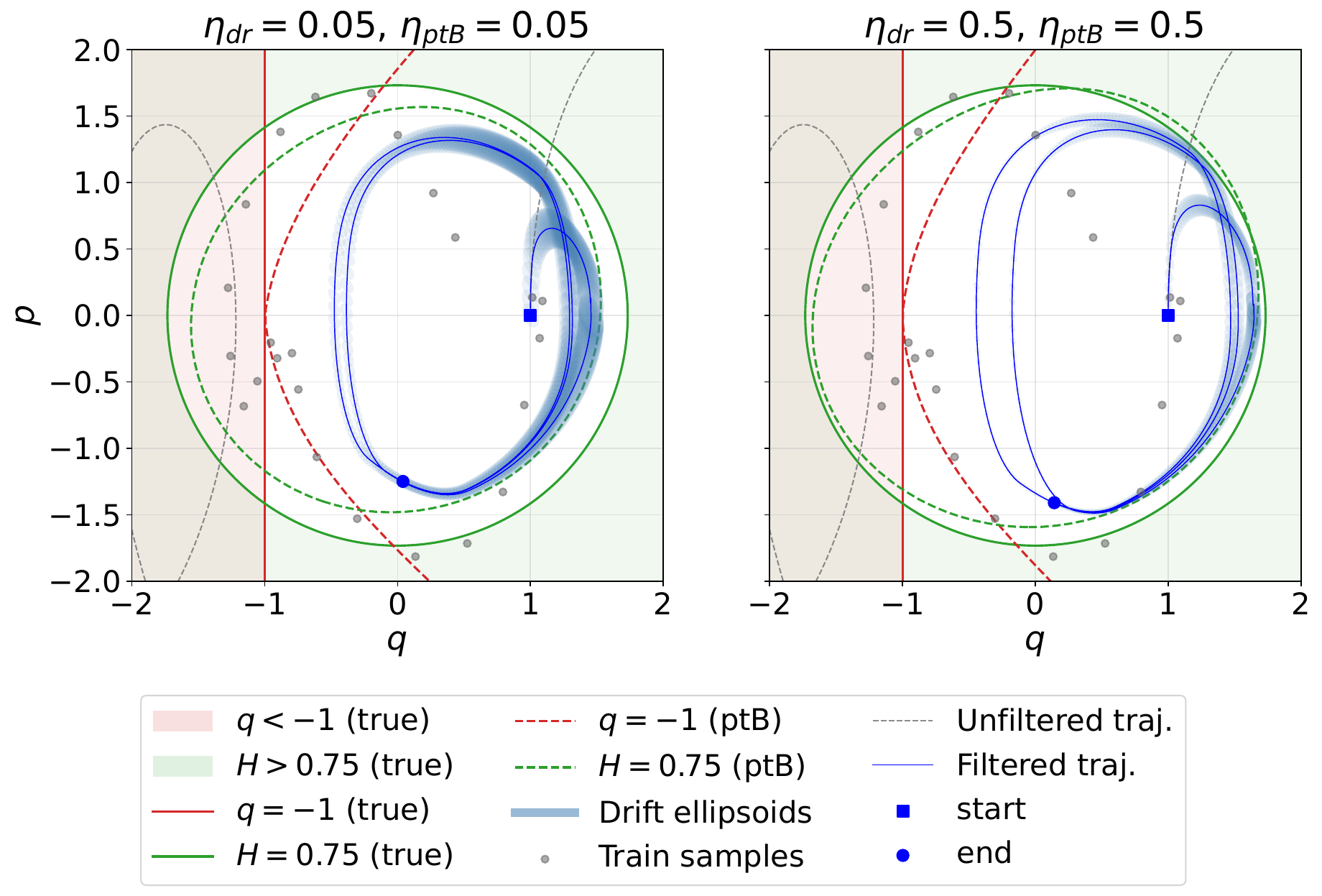}
    \caption{Phase portraits of the mass-spring system with kinematic ($q \geq -1$) and energy upper-bound ($H \leq 0.75$) barriers, trained on $K = 25$ noisy samples.
    Left: $(\eta_{\mathrm{dr}}, \eta_{\mathrm{ptB}}) = (0.05, 0.05)$.
    Right: $(0.5, 0.5)$.
    Solid contours denote true safe-set boundaries; dashed contours denote ptB barriers.
    Blue shaded regions show the drift credible ellipsoids along the filtered trajectory.}
    \label{fig:exp01_phase_portrait_ub}
\end{figure}
Fig.~\ref{fig:exp01_phase_portrait_ub} shows the phase-plane behavior for two credibility settings: a conservative pair $(\eta_{\mathrm{dr}}, \eta_{\mathrm{ptB}}) = (0.05, 0.05)$ and a less conservative pair $(0.5, 0.5)$.
In the conservative case (left panel), the ptB barriers (dashed contours) substantially enlarge the excluded region beyond the true unsafe set (solid contours), reflecting the wide posterior uncertainty from only $K = 25$ observations.
The drift credible ellipsoids (shaded blue) are correspondingly large, and the filtered trajectory (solid blue) remains well within the intersection of both ptB safe sets.
In the less conservative case (right panel), the ptB safe-set boundaries shrink toward their true counterparts: both the energy credible band and the drift ellipsoids contract, allowing the trajectory greater freedom.
The unfiltered nominal trajectory (dashed gray) violates both constraints in both panels, confirming that the safety filter intervenes.

\begin{figure}[t]
    \centering
    \includegraphics[width=\linewidth]{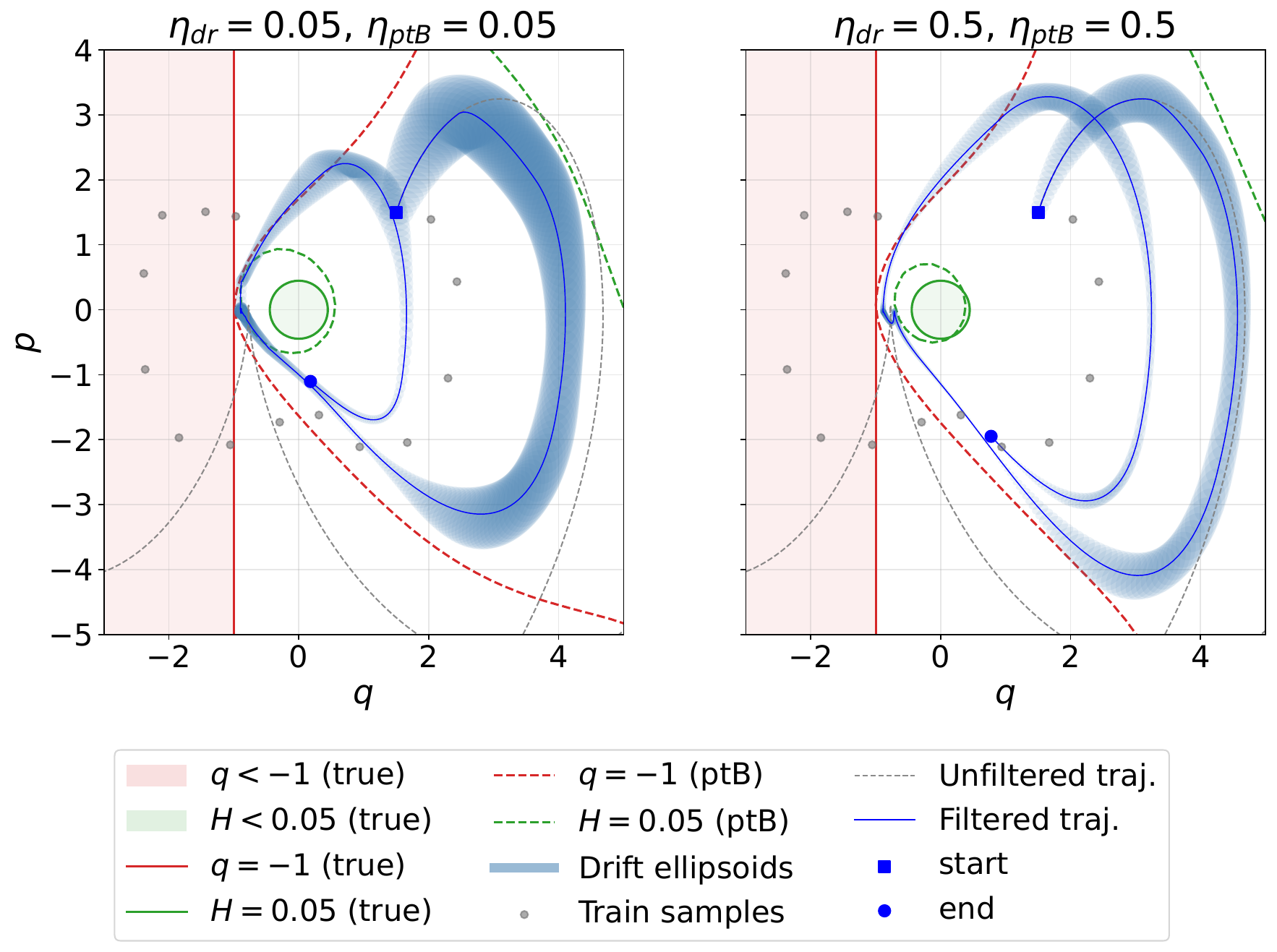}
    \caption{Phase portraits with kinematic ($q \geq -1$) and energy lower-bound ($H \geq 0.05$) barriers, $K = 15$.
    The conservative setting (left) yields wide ptB margins; the less conservative setting (right) shows smaller drift ellipsoids and a less conservative ptB barrier.
    It is worth noticing that the drift credible ellipsoids shrink near observed training samples in both figures.
    }
    \label{fig:exp01_phase_portrait_lb}
\end{figure}


Fig.~\ref{fig:exp01_phase_portrait_lb} shows the corresponding results for a lower energy-bound constraint $H(q,p) \geq 0.05$ with fewer data points ($K=15$).
Unlike the upper-bound case, the lower energy-bound constraint removes a neighborhood of the origin and therefore produces a non-convex safe set in the phase plane.
In this more restrictive geometry, the ptB-CBF filter remains feasible and enforces the prescribed safety specification over the full horizon.
As before, the conservative setting yields wider Bayesian margins, while the less conservative setting tracks the true boundary more closely.

\vspace{.5em}
\noindent\textit{Monte Carlo Credibility Validation:}
To assess the end-to-end safety guarantee of Theorem~\ref{th:bayesian-safety-pb-cbf-clean} empirically, we run independent train-filter-evaluate trials across a grid of credibility pairs $(\eta_{\mathrm{dr}}, \eta_{\mathrm{ptB}})$ on the mass-spring system.
Each trial draws a fresh noise realization, trains a GP from $K = 20$ observations, and constructs the ptB barriers and drift ellipsoid at the specified credibility levels.
The nominal controller is $u_{\mathrm{nom}}(t) = 3.0\,\cos(0.85\,t)$, chosen to violate safety constraints in the absence of filtering.
Parameters, initial conditions, and safety specifications are identical to the setup used in Fig.~\ref{fig:exp01_phase_portrait_ub}.

Fig.~\ref{fig:mc_empirical_credibility_compare_20K} shows the resulting Monte Carlo credibility validation.
The left panel plots the grouped empirical violation rate against the predicted bound $\eta_{\mathrm{dr}} + \eta_{\mathrm{ptB}}$.
The empirical violation rate remains monotone in the total credibility budget.
However, the discrepancy between prediction and observation is notably smaller in the high total-credibility-budget regime.
The right panel shows how the empirical violation rate changes over the full $(\eta_{\mathrm{dr}}, \eta_{\mathrm{ptB}})$ grid.
The plot shows that the violation rate is more sensitive to $\eta_{\mathrm{ptB}}$ than to $\eta_{\mathrm{dr}}$, revealing structure that is obscured when the results are collapsed onto the one-dimensional quantity $\eta_{\mathrm{dr}} + \eta_{\mathrm{ptB}}$.
Therefore, while the union bound captures the overall trend, the two uncertainty budgets enter asymmetrically in this example.


\begin{figure}[t]
    \centering
    \includegraphics[width=\linewidth]{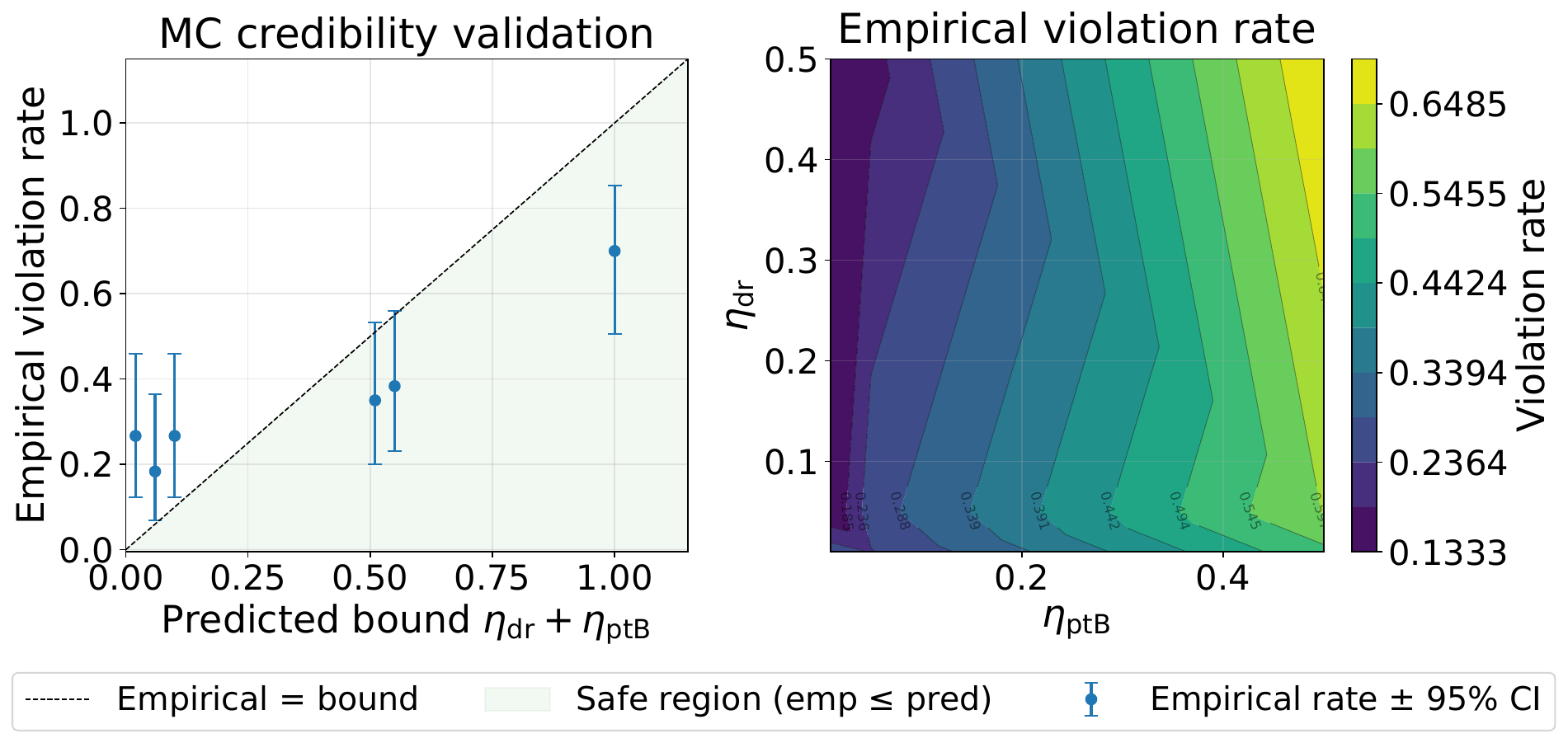}
    \caption{Monte Carlo credibility validation for the mass-spring system with $K=20$ training samples. Empirical violation rate (left), grouped by predicted bound $\eta_{\mathrm{dr}}+\eta_{\mathrm{ptB}}$, with $95\%$ confidence intervals across trials.
    The dashed diagonal denotes exact agreement, while the green region corresponds to empirically safe performance $\mathrm{emp}\leq \mathrm{pred}$. 
    Empirical violation-rate heatmap (right) over $\eta_{\mathrm{dr}}$ and $\eta_{\mathrm{ptB}}$ exposes the hidden structure beneath the left-panel trend, showing that the violation rate is noticeably more sensitive to $\eta_{\mathrm{ptB}}$ than to $\eta_{\mathrm{dr}}$. 
    }
    \label{fig:mc_empirical_credibility_compare_20K}
    \vspace{-0.5em}
\end{figure}

Fig.~\ref{fig:data_efficiency_violation_rate_summary} shows that the empirical violation rate decreases as more training data are used.
This trend indicates that the learned model becomes more reliable with increasing $K$, which in turn improves the practical safety performance of the ptB-CBF filter.
The shaded confidence band also suggests that the variation across credibility settings remains moderate and tends to shrink as the data become more informative.

\begin{figure}[t]
    \centering
    \includegraphics[width=\linewidth]{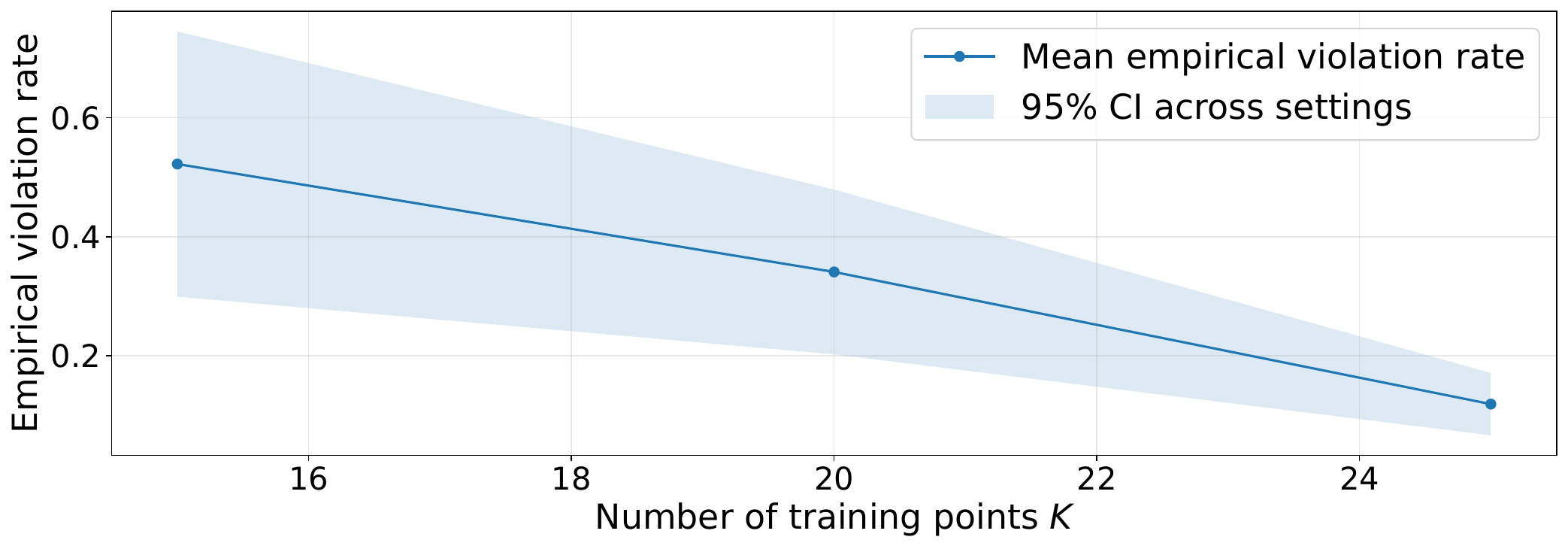}
    \caption{Empirical violation rate versus training-set size $K$, averaged across all combinations of credibility budgets.}
    \label{fig:data_efficiency_violation_rate_summary}
    \vspace{-1.5em}
\end{figure}

\section{Conclusion}\label{sec:conclusion}
We developed a port-transversal Bayesian control barrier function (ptB-CBF) framework for port-Hamiltonian systems with Hamiltonians learned from data. 
The key idea is to decompose the safety problem into two parts: a known structural component determined by the interconnection topology, and unknown energy values associated with the blanket subsystems. 
Given a calibrated posterior over the Hamiltonian---instantiated here via Gaussian process regression---we construct Bayesian barriers whose safe sets are high-probability inner approximations of the true allowable set.
In parallel, we derive a credible ellipsoid for the drift term to capture uncertainty in the control barrier function inequality. 
Since these two uncertainty sources are handled through separate credible sets, the resulting end-to-end safety guarantee is at least $1 - (\eta_{\mathrm{dr}} + \eta_{\mathrm{ptB}})$, with both confidence budgets chosen at design time. 
Then, we conduct numerical experiments on a mass-spring oscillator and a planar manipulator to validate the proposed ptB-CBF filter.
The filter produces conservative safety margins that account for uncertainty in both the posterior dynamics and the Hamiltonian surface.
We validate this credibility guarantee by comparing the empirical violation rate with the predicted bound.
In future work, we plan to extend the proposed framework to real-world robotic applications, validate the framework on hardware, and develop tighter probabilistic guarantees that reduce the conservatism of the current union-bound analysis.

\bibliographystyle{IEEEtran}
\bibliography{refs}

\section{Appendix}\label{sec:appendix}
Additional conceptual constructs are organized here.

\subsection{Bayesian-CBF Theory}\label{sec:bayesian-cbf-theory}
Learning-based CBF methods already enforce safety under model uncertainty by imposing high-probability variants of the standard CBF inequality, typically either robustly over a credible set for the learned drift, e.g., GP error bounds, or as chance constraints under a Bayesian dynamics model \cite{jagtap2020control,long2022safe,dhiman2021control}.
However, we have yet to see an explicit, model-agnostic layer as a reusable interfacing framework for later structured constructions.
To bridge this gap, we introduce the Bayesian-CBF framework and formulate the problem afterward.

We start by establishing a generic Bayesian formulation of the standard CBF theory.
Let the parameter space $\Theta \subset \mathbb{R}^p$ be a Borel subset with its Borel $\sigma$-algebra $\mathcal{B}(\Theta)$, and let $\pi(\cdot\mid\mathfrak{D})$ be a posterior on $\Theta$ with a dataset $\mathfrak{D} \coloneqq \{({x}_i, y_i)\}_{i=1}^N$.
Then,
\begin{equation*}
    \pi(\cdot\mid\mathfrak{D}): \mathcal{B}(\Theta)\to [0, 1]
\end{equation*}
is a probability measure with the tuple $(\Theta, \mathcal{B}(\Theta), \pi(\cdot\mid \mathfrak{D}))$ forming a probability space.
An example of the model parameter $\theta$ is the posterior mean $\mu_H$ defined in 
\cite[Eq.~(28)]{leung2025learning}.
A credible model set is defined as:
\begin{definition}[Credible model set]\label{df:cred-model-set}
    For a given credibility level \(1-\eta\in(0,1)\), a credible model set is a measurable subset \(\Theta_\eta\subset\Theta\) such that
    $
    \pi(\Theta_\eta\mid\mathfrak{D}) \ge 1-\eta.
    $
\end{definition}
For each $\theta\in \Theta$, an ODE with closed-loop dynamics:
\begin{equation}\label{eq:theta-ODE}
    \dot{x}(t) = f_\theta(x(t)) + g_\theta(x(t))u(x(t)),\ x(0)=x_0,
\end{equation}
is defined.
Let the solution to \eqref{eq:theta-ODE} be $x^\theta(t; x_0)$,
the solution $x^\theta(t; x_0)$
exists and depends continuously on $\theta$ under standard local Lipschitz and growth conditions.
The map $\theta\mapsto x^\theta(t;x_0)$ is Borel-measurable for each $t$, so for any closed set $\mathcal S\subset \mathcal{D}$, the event $\{\theta: x^\theta(t)\in\mathcal S\ \forall t\ge0\}$
is measurable. 
These regularities justify the definition of Bayesian forward invariance and the probability $\pi(\theta: x^\theta(t)\in\mathcal S\ \forall t\ge0\mid\mathfrak{D})$; see 
Appendix~\ref{sec:appendix-measure}.

\begin{definition}[Bayesian forward invariance]\label{df:B-forward-invar}
    Suppose the true dynamics are indexed by a random parameter \(\theta\in\Theta\) with posterior \(\pi(\cdot\mid\mathfrak{D})\). 
    For each \(\theta\), let \(x^\theta(t)\) be the closed-loop trajectory under policy \(u\).
    A set \(\mathcal S\) is \((1-\eta)\)-Bayesian forward invariant if, for all deterministic \(x_0\in\mathcal S\),
    \[
    \pi( \theta : x^\theta(t)\in\mathcal S \ \forall t\ge0 \mid \mathfrak{D}) \ge 1-\eta.
    \]
\end{definition}
Let \(h:\mathcal{D}\to\mathbb R\) be \(C^1\) and \(\mathcal S = \{x: h(x)\ge 0\}\), a uniform Bayesian-CBF is defined as:
\begin{definition}[Uniform Bayesian-CBF]\label{df:unif-B-cbf} 
We say that \(h\) is a \((1-\eta)\)-uniform Bayesian control barrier function (B-CBF) for the posterior \(\pi(\cdot\mid\mathfrak{D})\) if there exists an extended class-\(\mathcal K_\infty\) function \(\alpha(\cdot)\) such that, for every \(x\in \mathcal{D}\), 
\[ \sup_{u\in\mathcal U}\inf_{\theta\in\Theta_\eta} \bigl[ L_{f_\theta} h(x) + L_{g_\theta} h(x) u \bigr] \ge -\alpha(h(x)). 
\]
\end{definition}
Equivalently, the $\sup$-$\inf$ form of the uniform Bayesian-CBF condition can be written as follows if the supremum is attainable:
for all \(x\in \mathcal{D}\), there exists a \(u\in\mathcal U\) such that,
$$L_{f_\theta} h(x) + L_{g_\theta} h(x) u \ge -\alpha(h(x))\ \forall \theta\in \Theta_\eta.$$
The attainability of the supremum on $\mathcal{U}$ is guaranteed when $\mathcal{U}$ is nonempty, convex, and compact.
Now, let the Bayesian admissible control set be:
\begin{align*}
    &\mathcal{K}_{\mathrm{b.cbf}}(x)\coloneqq \\
    &\left\{u\in\mathcal{U}: L_{f_\theta} h(x) + L_{g_\theta} h(x) u \ge -\alpha \bigl(h(x)\bigr) \ \forall\theta\in \Theta_\eta\right\},
\end{align*}
and the Bayesian CBF-QP is:
\begin{equation*}
    \begin{aligned}
        u^*(x)=\ &\arg\min_{u\in\mathcal{U}}\ \|u-u_{\mathrm{nom}}(x)\|^2\\
        &\ \text{s.t.}\quad L_{f_\theta} h(x)+L_{g_\theta} h(x) u \ge -\alpha \bigl(h(x)\bigr),\ \forall\theta\in \Theta_\eta.
    \end{aligned}
\end{equation*}
Notice that the notion of uniform Bayesian-CBF is related to robust-CBF in the following way.
\begin{remark}[Connection to robust-CBFs]\label{rm:connection-bcbf-rcbf-cbvf}
Definition~\ref{df:unif-B-cbf} can be interpreted as a {robust}-CBF condition over a credible set of models.
For each state $x\in\mathcal D$, define the set-valued drift and input maps induced by the credible set:
\[
    \mathcal F_\eta(x)\coloneqq \{f_\theta(x):\theta\in\Theta_\eta\},\qquad
    \mathcal G_\eta(x)\coloneqq \{g_\theta(x):\theta\in\Theta_\eta\}.
\]
Then, the uniform B-CBF inequality is equivalently:
\[
    \sup_{u\in\mathcal U}\ \inf_{\substack{f\in\mathcal F_\eta(x)\\ g\in\mathcal G_\eta(x)}}
    \nabla h(x)^\top\bigl(f+g\,u\bigr)
    \ \ge\ -\alpha\bigl(h(x)\bigr),
\]
which is the standard robust-CBF form~\cite{xu2015robustness,jankovic2018robust,choi2021robust}, except that the uncertainty sets
$\mathcal F_\eta(x),\mathcal G_\eta(x)$ are {data driven} and {posterior calibrated}, i.e., they come from $\Theta_\eta$, rather than chosen a priori.
\end{remark}

We are ready to present the key theorem:
\begin{theorem}[Bayesian safety via B-CBFs]\label{th:B-forward-invar-via-B-cbf}
    Suppose:
    \begin{enumerate}[i)]
        \item \textit{CBF regularity}: $h\in C^1(\mathcal D)$ with $\nabla h(x)\neq 0$ for all $x\in\partial\mathcal S$, where $\mathcal S=\{x\in\mathcal D:h(x)\ge 0\}$,
        \item \textit{B-CBF}: The barrier function \(h\) is a \((1-\eta)\)-uniform Bayesian-CBF with credible set \(\Theta_\eta\),
        \item \textit{Lipschitz feedback}: The input \(u^*(x)\) is a Lipschitz continuous feedback with \(u^*(x)\in\mathcal K_{\rm b.cbf}(x)\) for all \(x\).
    \end{enumerate}
    Then,
    for any deterministic initial condition \(x_0\in\mathcal S\),
    \[
      \pi\bigl(\theta : x^\theta(t)\in\mathcal S\ \forall t\ge0 \mid \mathfrak{D}\bigr)  \ge  1-\eta,
    \]
    i.e.,\ \(\mathcal S\) is \((1-\eta)\)-Bayesian forward invariant.
\end{theorem}

\begin{proof}
Fix \(x_0\in\mathcal S = \{x:h(x)\ge 0\}\). For each \(\theta\in\Theta_\eta\), the uniform Bayesian-CBF property as in Definition~\ref{df:unif-B-cbf} and the choice \(u^*(x)\in\mathcal K_{\rm b.cbf}(x)\) imply that:
\[
    L_{f_\theta} h(x) + L_{g_\theta} h(x) u^*(x)
     \ge  -\alpha\bigl(h(x)\bigr)
    \qquad\forall x\in\mathcal D.
\]
Thus, for each fixed \(\theta\in\Theta_\eta\), \(h\) is a deterministic CBF for the control-affine system
\(
 \dot x = f_\theta(x)+g_\theta(x)u
\)
under the Lipschitz feedback \(u^*\). 
By the standard CBF forward-invariance theorem~\cite[Theorem~2]{ames2019control}, the closed-loop trajectory \(x^\theta(t;x_0)\) with \(x^\theta(0;x_0)=x_0\) satisfies:
\[
    x^\theta(t;x_0)\in\mathcal S
    \qquad\forall t\ge0,
    \ \forall \theta\in\Theta_\eta.
\]

Define the safety event:
\[
    \mathcal E(x_0)
    \coloneqq
    \bigl\{\theta\in\Theta : x^\theta(t;x_0)\in\mathcal S\ \forall t\ge0\bigr\}.
\]
By the measurability argument in Lemma~\ref{lem:measurable_safe_event}, \(\mathcal E(x_0)\in\mathcal B(\Theta)\), so \(\pi(\mathcal E(x_0)\mid\mathfrak D)\) is well defined. The deterministic argument above shows that
\(
   \Theta_\eta \subseteq \mathcal E(x_0)
\),
hence:
\begin{equation*}
    \begin{aligned}
        \pi\bigl(\theta : x^\theta(t;x_0)\in\mathcal S\ \forall t\ge0 \mid \mathfrak D\bigr)
         &= \pi\bigl(\mathcal E(x_0)\mid\mathfrak D\bigr)\\
         &\ge  \pi\bigl(\Theta_\eta\mid\mathfrak D\bigr)
          \ge  1-\eta,
    \end{aligned}
\end{equation*}
where the last inequality uses the definition of the credible model set \(\Theta_\eta\).
Since \(x_0\in\mathcal S\) was arbitrary, \(\mathcal S\) is \((1-\eta)\)-Bayesian forward invariant.
\end{proof}
The B-CBF formalism makes the recurring pattern in learning-based CBF methods~\cite{jagtap2020control,long2022safe,dhiman2021control} explicit and modular by defining Bayesian forward invariance and uniform B-CBFs over posterior credible model sets.
Altogether, these B-CBF related constructs provide a reusable framework that can be instantiated with different Bayesian learners and structured posteriors, e.g., Hamiltonian posteriors in pB-CBFs, without re-deriving the invariance logic each time.

\subsection{Measure Theoretic Setup for the Bayesian-CBF Theory}\label{sec:appendix-measure}
We provide the formality of the conditions that guarantee the existence and uniqueness of the solution $x^\theta(t)$ on $t\in [0, \infty)$ for each $\theta \in \Theta$ and measurability of infinite-horizon safety events.
\begin{assumption}[Solution regularity]\label{as:ode-regularity}
    Let the closed-loop vector field be $F_\theta(x) \coloneqq f_\theta(x) + g_\theta(x)u(x)$.
    Assume the following:
    \begin{enumerate}[i)]
        \item \textit{Continuity}: For each $\theta\in \Theta$, $F_\theta(\cdot)$ is uniformly continuous in $x \in \mathcal{D}$,
        \item \textit{Local Lipschitz}: For each compact ${\mathcal{K}} \subset \mathcal{D}$, there exists $L_{\mathcal{K}} > 0$ s.t.
        \begin{equation*}
            \|F_\theta(x) - F_\theta(y)\| \leq L_{\mathcal{K}}\|x - y\|\quad \forall x, y \in \mathcal{K}, \forall \theta \in \Theta,
        \end{equation*}
        \item \textit{Linear growth bound}: For some constant $a,b \geq 0$,
        \begin{equation*}
            \|F_\theta(x)\| \leq a\|x\| + b\quad \forall x \in \mathcal{D}, \forall \theta\in \Theta.
        \end{equation*}
    \end{enumerate}
\end{assumption}
Under Assumption~\ref{as:ode-regularity}, for each fixed $\theta\in\Theta_\eta$ and initial state $x_0\in D$, there is a unique maximal solution $x^\theta(\cdot;x_0)$ to the initial value problem (IVP).
Next, we are interested in formalizing the conditions for measurability of the map $\theta \mapsto x^\theta(t)$.

For each $\theta \in \Theta$ we consider the parametric ODE:
\begin{equation}\label{eq:parametric-ode}
  \dot{x}(t) = F(\theta, x(t)), 
  \qquad x(0) = x_0 \in \mathcal{D},
\end{equation}
with the corresponding $x^\theta(t; x_0) \in \mathcal{D}$.
In order to speak meaningfully about events of the form:
\[
  \{\theta \in \Theta : x^\theta(t; x_0) \in \mathcal{S}\},
  \qquad \mathcal{S} \subset \mathcal{D},
\]
we require that, for each fixed $t \ge 0$, the map
$
  \Theta \ni \theta \mapsto x^\theta(t; x_0) \in \mathcal{D}
$
is $\mathcal{B}(\Theta)$-measurable.

A standard result from parametric ODE theory of Carath\'eodory type ensures this measurability under mild regularity assumptions  on the vector field:
\begin{assumption}\label{ass:caratheodory}
Let $F : \Theta \times \mathcal{D} \to \mathbb{R}^n$ satisfy:
\begin{enumerate}[i)]
  \item $F$ is jointly continuous in $(\theta,x)$;
  \item $F$ is locally Lipschitz in $x$, uniformly in $\theta$ on compact subsets of $\Theta$. 
\end{enumerate}
\end{assumption}
Under Assumption~\ref{as:ode-regularity}~{and}~\ref{ass:caratheodory}, the initial value problem \eqref{eq:parametric-ode}
admits a unique local solution for each $(\theta,x_0) \in \Theta \times \mathcal{D}$, and the solution map:
\[
  (\theta,x_0) \longmapsto x(t;\theta,x_0)
\]
is continuous for every fixed $t$.
In particular, fixing $x_0$ and $t$, the map:
\[
  \Theta \ni \theta \longmapsto x^\theta(t; x_0)
\]
is continuous, hence Borel measurable, on any subset $\Theta_\eta \subset \Theta$ where the above conditions hold.

Consequently, for any Borel set $\mathcal{S} \subset \mathcal{D}$ and any $t \ge 0$,
\[
  \{\theta \in \Theta : x^\theta(t; x_0) \in \mathcal{S}\} \in \mathcal{B}(\Theta).
\]
Therefore, one can meaningfully consider pathwise events such as:
\[
  \bigl\{\theta \in \Theta : x^\theta(t; x_0) \in \mathcal{S}\  \forall t \ge 0\bigr\},
\]
once an appropriate $\sigma$-algebra on the path space is specified.
This provides the basic measurability setup for a Bayesian model and allows us to establish the measurability of infinite-horizon safety events.

\begin{lemma}[Measurability of the infinite-horizon safety event]\label{lem:measurable_safe_event}
Let $\mathcal{S} = \{x \in \mathcal{D} : h(x) \ge 0\}$ be a closed safe set and fix $x_0 \in \mathcal{D}$.  
Under Assumptions~\ref{as:ode-regularity} and~\ref{ass:caratheodory}, for each $t \ge 0$ the solution map
\[
  \Theta \ni \theta \longmapsto x^\theta(t; x_0) \in \mathcal{D}
\]
is Borel-measurable, and the event
\[
  A_\infty 
  \coloneqq 
  \bigl\{\theta \in \Theta : x^\theta(t; x_0) \in \mathcal{S} \ \forall t \ge 0\bigr\}
\]
belongs to $\mathcal{B}(\Theta)$. Consequently, the probability
\[
  \pi\bigl(A_\infty \mid \mathfrak{D}\bigr)
  =
  \pi\bigl(\theta : x^\theta(t; x_0) \in \mathcal{S}\ \forall t \ge 0  \bigm|  \mathfrak{D}\bigr)
\]
is well-defined.
\end{lemma}

\begin{proof}
Fix $t \ge 0$ and define
\[
  A_t \coloneqq \{\theta \in \Theta : x^\theta(t; x_0) \in \mathcal{S}\}.
\]
By Assumption~\ref{ass:caratheodory}, the map $\theta \mapsto x^\theta(t; x_0)$ is continuous (hence Borel-measurable) for each fixed $t$.  
Since $\mathcal{S}$ is closed, the indicator of $\mathcal{S}$ is Borel-measurable, so $A_t \in \mathcal{B}(\Theta)$ for every $t \ge 0$.

We now show that the infinite-horizon safety event is measurable.  
Define $A_\infty \coloneqq \{\theta \in \Theta : x^\theta(t; x_0) \in \mathcal{S} \ \forall t \ge 0\}$.
Because each trajectory $t \mapsto x^\theta(t; x_0)$ is continuous in $t$ and $\mathcal{S}$ is closed, we have
\[
  x^\theta(t; x_0) \in \mathcal{S} \ \forall t \ge 0
  \quad\Longleftrightarrow\quad
  x^\theta(t; x_0) \in \mathcal{S} \ \forall t \in \mathbb{Q}_{\ge 0},
\]
where $\mathbb{Q}_{\ge 0}$ denotes the set of nonnegative rationals.  
Indeed, if a continuous trajectory ever leaves $\mathcal{S}$ at some time $t^\ast$, then by continuity it must cross the boundary of $\mathcal{S}$ at times arbitrarily close to $t^\ast$, and in particular at some rational time.

Therefore,
\[
  A_\infty
  =
  \bigcap_{q \in \mathbb{Q}_{\ge 0}}
  \{\theta \in \Theta : x^\theta(q; x_0) \in \mathcal{S}\}
  =
  \bigcap_{q \in \mathbb{Q}_{\ge 0}} A_q.
\]
The index set $\mathbb{Q}_{\ge 0}$ is countable, and each $A_q \in \mathcal{B}(\Theta)$, so $A_\infty$ is a countable intersection of measurable sets and hence belongs to $\mathcal{B}(\Theta)$.

It follows that $\pi(A_\infty \mid \mathfrak{D})$ is well-defined, which provides the measure-theoretic justification for the Bayesian forward-invariance probability
  $\pi\bigl(\theta : x^\theta(t; x_0) \in \mathcal{S} \ \forall t \ge 0  \bigm|  \mathfrak{D}\bigr)$.
\end{proof}

\subsection{CBF Feasibility under Bounded Inputs}
\label{sec:cbf-feasibility}
The following assumptions developed in~\cite{leung2026port}, 
ensure
that the CBF inequality for~$h_\gamma$ is feasible under bounded
inputs.

\begin{assumption}[Port-insulated barrier~{\cite[Asm.~III-B.1]{leung2026port}}]%
\label{ass:port-insulated}
  $\partial\mathcal{B} \subseteq \mathcal{P}$.
\end{assumption}

\begin{assumption}[Blanket coercivity~{\cite[Asm.~III-B.2]{leung2026port}}]%
\label{ass:blanket-coercivity}
  There exists $\sigma_{\mathcal{B}} > 0$ such that
  $\|L_G h_\gamma(x)\|
  \geq (\sigma_{\mathcal{B}}/\gamma)\,
  \|e_{\partial\mathcal{B}}(x)\|$
  for all $x \in \mathcal{D}$,
  where
  $\|e_{\partial\mathcal{B}}(x)\|
  \coloneqq
  (\sum_{j\in\partial\mathcal{B}}
  \|\nabla H_j(x_j)\|^2)^{1/2}$
  is the collective blanket effort norm.
\end{assumption}


\begin{assumption}[Dominated degeneracy~{\cite[Asm.~III-B.4]{leung2026port}}]%
\label{ass:dom-degen}
There exists an open neighborhood~$U$ of
$\mathcal{Z}_{\partial\mathcal{B}} \cap
\partial\mathcal{S}_\gamma$
in~$\partial\mathcal{S}_\gamma$ and $c_0 \geq 0$ such that:
\begin{equation}\label{eq:dom-degen}
(-L_f h_\gamma(x))^+ \leq c_0 \|e_{\partial \mathcal B}(x)\|
\qquad \forall\, x \in U.
\end{equation}
\end{assumption}

\begin{assumption}[Compact design safety
  boundary~{\cite[Asm.~5]{leung2026port}}]%
\label{ass:compact}
  The set $\partial\mathcal{S}_\gamma$ is compact.
\end{assumption}

\begin{theorem}[CBF feasibility~{\cite[Thm.~2]{leung2026port}}]%
\label{thm:pt-feasibility}
  Under
  Assumptions~\ref{ass:structural}--\ref{ass:compact},
  the worst-case authority ratio
  \begin{equation}\label{eq:authority-bound}
    \bar{u}_\gamma
    \coloneqq
    \sup_{x \in \partial\mathcal{S}_\gamma
          \setminus\mathcal{Z}_{\partial\mathcal{B}}}
    \frac{(-L_f h_\gamma(x))^+}{\|L_G h_\gamma(x)\|}
  \end{equation}
  is finite.
  Hence, if
  $\mathcal{U} \supseteq \{u : \|u\| \leq \bar{u}_\gamma\}$,
  then the CBF condition for~$h_\gamma$ is feasible on
  $\partial\mathcal{S}_\gamma$, and~$\mathcal{S}_\gamma$ is
  forward invariant under any safety filter enforcing it.
\end{theorem}

\subsection{Mechanical Instantiation}
\label{sec:app:mech-inst}
For a fully actuated mechanical system in canonical coordinates
$(q,p)$ with Hamiltonian $H(q,p) = T(q,p) + V(q)$,
kinetic energy $T = \tfrac{1}{2}p^\top M(q)^{-1}p$, and
full-rank input matrix~$B$, the port-Hamiltonian
form~\eqref{eq:true-dynamics} has
$J = \bigl[\begin{smallmatrix} 0 & I \\ -I & 0
\end{smallmatrix}\bigr]$,
$R = \bigl[\begin{smallmatrix} 0 & 0 \\ 0 & D
\end{smallmatrix}\bigr]$,
$G = \bigl[\begin{smallmatrix} 0 \\ B
\end{smallmatrix}\bigr]$.
When $M(q)$ depends on configuration, separability is restored
by the screw-momenta
lifting~\cite[Ch.~3.2]{duindam2009modeling}.

\begin{corollary}[Energy-aware barrier~{\cite[Cor.~2]{leung2026port}}]%
\label{cor:mechanical}
  For a configuration constraint $\varphi(q) \geq 0$, the
  barrier-insulating blanket is
  $\partial\mathcal{B} = \{p\}$ with
  $\partial\mathcal{B} \cap \mathcal{P} = \{p\}$, and the
  port-transversal barrier reduces to:
  \begin{equation}\label{eq:mech-barrier}
    h_\gamma(q,p) = \varphi(q) - \frac{1}{\gamma}\,T(q,p),
  \end{equation}
  with $L_G h_\gamma = -\tfrac{1}{\gamma}\,\dot{q}^\top B$,
  nonzero whenever $\dot{q} \neq 0$.
  Assumptions~\ref{ass:port-insulated}--\ref{ass:compact} are
  satisfied, and $\mathcal{S}_\gamma$ is forward invariant
  under the safety filter of Theorem~\ref{thm:pt-feasibility}.
\end{corollary}

Corollary~\ref{cor:mechanical} recovers the energy-aware barrier
of~\cite{singletary2021safety} as a special case.
For the Bayesian development that follows, the key observation
is that evaluating~\eqref{eq:mech-barrier} requires the kinetic
energy $T(q,p)$, which in turn depends on the unknown
Hamiltonian.
More generally, evaluating any port-transversal
barrier~\eqref{eq:pt-barrier-blanket} requires the shifted storages
$\{\bar{H}_j(x_j)\}_{j \in \partial\mathcal{B} \cap \mathcal{P}}$,
which are the scalar quantities whose posterior uncertainty the
next section addresses.

\end{document}